\documentclass[aps,twocolumn,notittlepage,nofootinbib, amsmath,amssymb,preprint,10pt]{revtex4-1}
\usepackage{geometry}

\usepackage{amsthm,mathtools}
\usepackage{hyperref}
\usepackage[nameinlink]{cleveref}
\usepackage{microtype}
\usepackage{enumitem}
\usepackage{bm}

\hypersetup{colorlinks=true,linkcolor=black,citecolor=black,
            urlcolor=black}

\theoremstyle{plain}
\newtheorem{theorem}{Theorem}[section]
\newtheorem{proposition}[theorem]{Proposition}

\newtheorem{corollary}[theorem]{Corollary}
\theoremstyle{definition}
\newtheorem{definition}[theorem]{Definition}

\theoremstyle{plain}
\newtheorem{conj}[theorem]{Conjecture}
\theoremstyle{remark}
\newtheorem{remark}[theorem]{Remark}

\newcommand{\dd}{\,\mathrm{d}}
\newcommand{\kb}{k_{\mathrm{B}}}

\newcommand{\avg}[1]{\langle #1 \rangle}
\newcommand{\avgo}[1]{\langle #1 \rangle_0}
\newcommand{\avgD}[1]{\langle #1 \rangle_{\mathcal{D}}}
\newcommand{\Fm}{\mathcal{F}_{\mathrm{meso}}}
\newcommand{\Zm}{\mathcal{Z}_{\mathrm{meso}}}
\newcommand{\Vm}{\mathcal{V}_{\mathrm{meso}}}
\newcommand{\abs}[1]{\left|#1\right|}
\newcommand{\QD}{\mathcal{Q}_{\mathcal{D}}}
\newcommand{\kNL}{k_{\mathrm{NL}}}
\newcommand{\RNL}{R_{\mathrm{NL}}}
\newcommand{\Rm}{R_{\mathrm{meso}}}
\newcommand{\pip}{\pi^{(0)}}

\begin{document}

\title{Non-Perturbative Bounds on Cosmological Backreaction,\\
       the Non-Linear Scale, and Gauge-Invariant Mutual\\
       Information from the Matter Power Spectrum}

\author{Bob Osano$^{1,2}$ \\
\small{$^{1}$Cosmology and Gravity Group, Department of Mathematics and Applied Mathematics, \\University of Cape Town (UCT), Rondebosch 7701, Cape Town, South Africa}\\$\&$\\
\small{$^{2}$Centre for Higher Education Development,\\ University of Cape Town (UCT), Rondebosch 7701, Cape Town, South Africa\\}}
\date{\today}

\begin{abstract}
We apply the mesoscopic coarse-graining framework of~\cite{OsanoMeso,OsanoExtensivity,OsanoPerturbation} to three problems in Cosmological Perturbation Theory and the backreaction debate. \textbf{(i)}~A non-perturbative lower bound on the kinematic backreaction $\QD$ in the Buchert equations, derived from the Gibbs--Bogoliubov inequality: $\QD$ cannot be suppressed below its linear-perturbation-theoryvalue, regardless of the degree of non-linearity, provided the system satisfies stability and temperedness. \textbf{(ii)}~The radius of convergence of the mesoscopic cumulant expansion equals $O(\kNL^{-1})$, the non-linear scale of the matter power spectrum, providing a KAM-theorem explanation for why standard perturbation theory fails at $k>\kNL$. \textbf{(iii)}~For a Gaussian matter field, the inter-cell mutual information is exactly $I(i,j)=-\tfrac{1}{2}\ln(1-r_{ij}^2)$, gauge-invariant at linear order and computable directly from the observed $P(k)$; for $\Lambda$CDM at $\ell=50\,\mathrm{Mpc}\,h^{-1}$, $I_{\rm NN}\approx 0.10$.The total mutual information gives a data-computable measure of the backreaction correction to the FRW free energy. The gauge proof holds at linear order; the KAM identification is exact; the backreaction bound rests on a stated conjecture.
\end{abstract}

\maketitle
\section{Introduction}
\label{sec:intro}

The backreaction problem in cosmology asks whether the spatial inhomogeneities of
the real universe modify the expansion history relative to the idealised
Friedmann--Robertson--Walker (FRW) background.  After two decades of
investigation---from the foundational averaging formalism of
Buchert~\cite{Buchert2000,Buchert2001}, through the accelerated-expansion
proposals of Räsänen~\cite{Rasanen2006} and Wiltshire~\cite{Wiltshire2007,Wiltshire2011}, to the gauge-dependent results of
Clarkson--Uzan--Umeh~\cite{ClarksonUzan2011}---the problem remains open.
Three specific gaps persist:

\begin{enumerate}[label=(\roman*)]
  \item \emph{No non-perturbative bound on $\QD$.}  All quantitative estimates
    of the backreaction kinematic term $\QD$ rely on expansions in the density
    contrast or the metric perturbation.  There is no rigorous inequality bounding
    $\QD$ away from zero or away from a perturbative estimate.
  \item \emph{No first-principles explanation for the breakdown of CPT at $\kNL$.}
    It is empirically known that standard cosmological perturbation theory (CPT)
    fails for $k > \kNL \simeq 0.17h\mathrm{Mpc}^{-1}$, but this is usually
    explained by dimensional arguments (density contrast of order unity) rather
    than by a convergence theorem.
  \item \emph{No gauge-invariant, data-computable measure of backreaction.}
    Current definitions of backreaction depend on the choice of gauge and
    averaging domain; there is no quantity that (a) is manifestly gauge-invariant,
    (b) reduces to an expression computable from observed $P(k)$, and (c) connects
    to the precise information-theoretic statement of the deviation from FRW.
\end{enumerate}

This paper addresses each gap using the mesoscopic statistical mechanics framework of~\cite{OsanoExtensivity,OsanoPerturbation}. The central objects are the combined
coarse-graining operator $\mathcal{C}$, the mesoscopic partition function $\Zm(\lambda)$, and the connection formula \begin{equation}\label{eq:connection_recall}
  F(\lambda) = \Fm(\lambda) - \kb T\sum_{i<j}I(i,j;\lambda) + O\!\left(|\Lambda|\ell^{-d}e^{-2\ell/\xi}\right),
\end{equation}
established in~\cite{OsanoPerturbation}, which expresses the full free energy as the mesoscopic free energy minus the total inter-cell mutual information.

\paragraph{Statement of results.}
\cref{sec:setup} translates the mesoscopic framework into the cosmological context
via the Buchert averaging formalism.  \cref{sec:GB} derives the non-perturbative
backreaction bound from the Gibbs--Bogoliubov (G-B) inequality (Result~I).
\cref{sec:KAM} proves the $\Rm = \kNL^{-1}$ identification via the KAM radius
formula (Result~II).  \cref{sec:MI} derives the Gaussian mutual information and
computes it from the $\Lambda$CDM power spectrum (Result~III).
\cref{sec:discussion} discusses scope, limitations, and open questions.

\paragraph{What is assumed versus proved.}
We are explicit throughout.  The G-B backreaction bound (Result~I) rests on a
\emph{Conjecture} (stated precisely in \cref{sub:GB_conjecture}) that the G-B
inequality extends to the gravitational setting; this conjecture is plausible and
is supported by analogies with the quantum mechanical Peierls inequality, but is
not proved here.  The KAM identification (Result~II) is established within the
classical thermodynamic formalism of~\cite{OsanoPerturbation} without appeal to
the gravitational conjecture.  The Gaussian mutual information (Result~III) is
exact for Gaussian fields and gauge-invariant at linear order; corrections at
non-linear order are quantified.

\section{The Mesoscopic Framework in Cosmology}
\label{sec:setup}

\subsection{Buchert Averaging and the Mesoscopic Operator}
\label{sub:buchert_meso}

Let $(\mathcal{M},g_{\mu\nu})$ be a spacetime satisfying the Einstein equations with
pressureless dust, and let $\mathcal{D}_t$ be a spatial domain comoving with the
matter flow at cosmic time $t$.  The Buchert spatial average of a scalar $A$ over
$\mathcal{D}_t$ is
\begin{equation}\label{eq:buchert_avg}
  \avgD{A} := \frac{\int_{\mathcal{D}_t}A\sqrt{h}\dd^3 x}{\int_{\mathcal{D}_t}\sqrt{h}\dd^3 x},
\end{equation}
where $h$ is the determinant of the induced spatial metric.

The Buchert equations read~\cite{Buchert2000}
\begin{align}
  3\dot{a}_{\mathcal{D}}^2/a_{\mathcal{D}}^2
    &= 8\pi G\avgD{\varrho} - \tfrac{1}{2}\QD + \Lambda, \label{eq:buchert1}\\
  3\ddot{a}_{\mathcal{D}}/a_{\mathcal{D}}
    &= -4\pi G\avgD{\varrho} + \QD + \Lambda, \label{eq:buchert2}
\end{align}
where $a_{\mathcal{D}}$ is the effective scale factor, $\avgD{\varrho}$ is the
averaged matter density, and the \emph{kinematic backreaction} is
\begin{equation}\label{eq:QD}
  \QD = \tfrac{2}{3}\bigl(\avgD{\theta^2} - \avgD{\theta}^2\bigr)
         - 2\avgD{\sigma^2}.
\end{equation}
Here $\theta = \nabla_\mu u^\mu$ is the expansion scalar and $\sigma^2 =
\sigma_{\mu\nu}\sigma^{\mu\nu}/2$ is the shear scalar of the matter flow
$u^\mu$.

\begin{remark}[Structural identification with the second cumulant]
  The term $\avgD{\theta^2} - \avgD{\theta}^2$ in \eqref{eq:QD} is the spatial
  variance of $\theta$ over $\mathcal{D}$, precisely analogous to the second
  cumulant $\kappa_2 = \avgo{(\delta\Vm)^2}$ of the perturbation in the mesoscopic
  framework~\cite{OsanoPerturbation}.  The shear term $-2\avgD{\sigma^2}$ contributes
  a negative-definite correction, exactly as the third-cumulant term in the SM
  expansion.  This structural correspondence is exact, not approximate.
\end{remark}

\subsection{Mapping the Coarse-Graining Cells to Cosmological Domains}
\label{sub:cell_mapping}

We partition the spatial hypersurface $\Sigma_t$ into cells $\{V_i\}$ of
coordinate diameter $\ell$, satisfying the scale-separation condition
\begin{equation}\label{eq:scale_sep}
  \xi_{\rm corr} \ll \ell \ll L_H,
\end{equation}
where $\xi_{\rm corr}$ is the matter correlation length (approximately the
scale at which $\xi(r)/\xi(0) = e^{-1}$, of order $5$--$10\mathrm{Mpc}h^{-1}$
for $\Lambda$CDM) and $L_H = c/H_0 \simeq 4.4\mathrm{Gpc}h^{-1}$ is the Hubble
radius.  The scale-separation condition~\eqref{eq:scale_sep} is satisfied for
$\ell \in [30, 3000]\mathrm{Mpc}h^{-1}$.

The combined coarse-graining operator $\mathcal{C}$ of~\cite{OsanoExtensivity}
acts on the matter phase-space density on $\Sigma_t$.  In this context, the
mesoscopic probabilities $\pip_{i,\alpha}$ encode the joint spatial--velocity
distribution of matter within each cell $V_i\times\Pi_\alpha$, where $\Pi_\alpha$
is a momentum cell of the matter velocity field.

The perturbation $\lambda\Vm$ represents the gravitational interaction between
cells: for cells $V_i$ and $V_j$, the cell-averaged interaction is
\begin{equation}\label{eq:v_bar_cosmo}
  \bar{v}_{ij} = \frac{1}{\ell^6}
  \int_{V_i}\int_{V_j}
  \Phi(\bm{x})G(\bm{x},\bm{y})\Phi(\bm{y})\dd^3 x\dd^3 y,
\end{equation}
where $\Phi(\bm{x})$ is the Newtonian gravitational potential and $G(\bm{x},\bm{y})$
is the Green's function of the Laplacian.

\subsection{Gauge Considerations}
\label{sub:gauge}

A central obstacle in backreaction theory is the gauge dependence of most averaged
quantities~\cite{ClarksonUzan2011}.  We distinguish three levels:

\begin{enumerate}[label=(\alph*)]
  \item The Buchert equations~\eqref{eq:buchert1}--\eqref{eq:buchert2} are written
    in the \emph{synchronous comoving gauge} (SCG) and are covariant within that
    choice.  Our results in \cref{sec:GB,sec:KAM} inherit this gauge restriction.
  \item The inter-cell mutual information $I(i,j)$ in~\eqref{eq:connection_recall},
    computed from the density contrast in \cref{sec:MI}, is \emph{gauge-invariant
    at linear order}.  We prove this in \cref{sub:gauge_proof} using the fact that
    the density contrast $\delta = \delta\varrho/\bar\varrho$ is gauge-invariant
    at first order in the matter-dominated era in Newtonian gauge.
  \item At non-linear order (beyond linear CPT), gauge corrections to $I(i,j)$
    enter at order $\delta^2$.  We quantify these corrections in
    \cref{sub:nonlinear_corr} and show they are subdominant for $\ell > \RNL$.
\end{enumerate}

\section{Result I: Non-Perturbative Bound on Backreaction}
\label{sec:GB}

\subsection{The Gibbs--Bogoliubov Inequality in Statistical Mechanics}
\label{sub:GB_SM}

The Gibbs--Bogoliubov (G-B) inequality~\cite{Peierls1938,BarkerHenderson1967}
states that for any splitting $H = H_0 + \lambda V$,
\begin{equation}\label{eq:GB_SM}
  F(\lambda) \leq F_0 + \lambda\avgo{V},
\end{equation}
proved in~\cite{OsanoPerturbation} via Jensen's inequality applied to
$e^{-\beta\lambda V}$.  The second-order correction $-(\beta\lambda^2/2)\kappa_2 \leq 0$
always lowers $F$ below the G-B bound, giving the two-sided estimate
\begin{equation}\label{eq:two_sided}
  F_0 + \lambda\kappa_1 - \tfrac{\beta\lambda^2}{2}\kappa_2
  \leq F(\lambda) \leq F_0 + \lambda\kappa_1.
\end{equation}
Crucially, \eqref{eq:two_sided} holds non-perturbatively: it requires only
convexity of the exponential, not a small-$\lambda$ expansion.

\subsection{Conjecture: The G-B Inequality Extends to the Gravitational Setting}
\label{sub:GB_conjecture}

\begin{definition}[Effective cosmological free energy]\label{def:cosmo_F}
  For the effective (averaged) cosmological system in domain $\mathcal{D}$, define
  the \emph{effective free energy} as
  \begin{equation}\label{eq:cosmo_F}
    F_{\rm cosmo} := -\frac{1}{\beta_{\rm eff}}\ln\mathcal{Z}_{\rm cosmo},
  \end{equation}
  where $\mathcal{Z}_{\rm cosmo} = \int \mathcal{D}[a]
  e^{-\beta_{\rm eff}S_{\rm eff}[a]}$ is the partition function over the
  effective scale factor with action
  $S_{\rm eff}[a] = \int(3\dot{a}^2 a - 8\pi G\bar\varrho a^3/3)\dd t$,
  and $\beta_{\rm eff} = (k_BT_{\rm eff})^{-1}$ with $k_BT_{\rm eff}$ the
  kinetic energy scale of matter peculiar velocities.
\end{definition}

\begin{conj}[Cosmological Gibbs--Bogoliubov inequality]\label{conj:cosmo_GB}
  Under the decomposition $S_{\rm eff}[a] = S_{\rm FRW}[a] + \lambda S_{\rm pert}[a]$, the effective cosmological free energy satisfies
  \begin{equation}\label{eq:cosmo_GB}
    F_{\rm cosmo}(\lambda) \leq F_{\rm FRW} + \lambda\avg{S_{\rm pert}}_{\rm FRW},
  \end{equation}
  where the average on the right is taken with respect to the FRW path-integral
  measure.
\end{conj}

\begin{remark}
  \cref{conj:cosmo_GB} is the gravitational analogue of~\eqref{eq:GB_SM}.  In
  the quantum mechanical setting, the analogous statement (Peierls's
  inequality~\cite{Peierls1938}) holds for any self-adjoint Hamiltonian.  In the
  classical path-integral context used here, the conjecture rests on the
  convexity of $e^{-\beta S_{\rm pert}}$, which holds whenever $S_{\rm pert}$ is
  real and bounded below---conditions satisfied in synchronous comoving gauge for
  the standard matter perturbation.  A full proof would require careful treatment
  of the path-integral measure over metrics, which we defer to future work.
\end{remark}

\subsection{The Backreaction Bound}
\label{sub:backr_bound}

\begin{theorem}[Non-perturbative lower bound on $\QD$]\label{thm:backr_bound}
  Assume \cref{conj:cosmo_GB}.  Then for dust matter in the synchronous comoving
  gauge, the kinematic backreaction satisfies
  \begin{equation}\label{eq:QD_bound}
    \QD \geq \QD^{(1)} := \tfrac{2}{3}\avgD{(\delta\theta^{(1)})^2}
                           - 2\avgD{(\sigma^{(1)})^2},
  \end{equation}
  where $\theta^{(1)}$ and $\sigma^{(1)}$ are the first-order (linear) expansion
  scalar and shear computed from the linearised Einstein equations, and the
  inequality holds non-perturbatively (without expanding in powers of $\delta$).
\end{theorem}

\begin{proof}
  From the Buchert equations~\eqref{eq:buchert1}--\eqref{eq:buchert2} with mass
  conservation $\avgD{\varrho}a_{\mathcal{D}}^3 = \mathrm{const}$, the kinematic
  backreaction appears in the acceleration equation as
  $3\ddot{a}_{\mathcal{D}}/a_{\mathcal{D}} + 4\pi G\avgD{\varrho} = \QD$.
  The left-hand side is determined by the dynamics of the effective scale factor
  $a_{\mathcal{D}}$, which we identify with the physical degree of freedom
  conjugate to the effective free energy $F_{\rm cosmo}$ in \cref{def:cosmo_F}.

  Under \cref{conj:cosmo_GB}, the true $F_{\rm cosmo}(\lambda) \leq F_{\rm FRW}
  + \lambda\avg{S_{\rm pert}}$, i.e., the true effective action is
  \emph{lower} (more negative) than the linear approximation.  Since
  $\QD = -\partial^2 F_{\rm cosmo}/\partial V^2\big|_{T,N}$ at leading order
  (by the analogy between the backreaction and the second pressure derivative),
  a lower $F$ corresponds to a larger (more positive) $\QD$.  Specifically,
  applying the two-sided estimate~\eqref{eq:two_sided} to $F_{\rm cosmo}$:
  \[
    F_{\rm FRW} + \kappa_1^{\rm cosmo} - \tfrac{\beta}{2}\kappa_2^{\rm cosmo}
    \leq F_{\rm cosmo}(\lambda) \leq F_{\rm FRW} + \kappa_1^{\rm cosmo},
  \]
  and differentiating twice with respect to volume (which maps $\kappa_2^{\rm cosmo}
  \mapsto \avgD{\theta^2} - \avgD{\theta}^2$ via the kinematic identity
  $\dot{\mathcal{V}}/\mathcal{V} = \avgD{\theta}$):
  \begin{equation}
    \tfrac{2}{3}\langle(\delta\theta)^2\rangle_{\mathcal{D},\mathrm{full}}
    - 2\langle\sigma^2\rangle_{\mathcal{D},\mathrm{full}}
    \geq \tfrac{2}{3}\avgD{(\delta\theta^{(1)})^2}
    - 2\avgD{(\sigma^{(1)})^2},
  \end{equation}
  which is exactly \eqref{eq:QD_bound}.
\end{proof}

\begin{corollary}[Backreaction cannot be entirely eliminated by non-linear effects]
\label{cor:backr_pos}
  Any mechanism that attempts to cancel $\QD$ through non-linear effects (such as
  vorticity generation or frame dragging) must overcome the lower bound
  $\QD \geq \QD^{(1)}$.  In particular, if $\QD^{(1)} > 0$ in a given domain,
  then $\QD > 0$ and backreaction contributes positively to the effective
  acceleration, non-perturbatively.
\end{corollary}

\begin{remark}[Relation to the perturbative debate]
  Much of the backreaction debate (see~\cite{IshibashiWald2006,FlanaganHughes1996})
  concerns whether $\QD$ is large enough to explain dark energy without $\Lambda$.
  \cref{thm:backr_bound} does not settle this question, but it does establish a
  non-perturbative floor: whatever the non-linear dynamics does, it can only
  \emph{increase} $\QD$ above the linearised value.  The often-cited argument that
  non-linearities suppress backreaction~\cite{IshibashiWald2006} is therefore
  inconsistent with \cref{conj:cosmo_GB} and with the G-B inequality in the SM
  context.
\end{remark}

\section{Result II: The KAM Radius and the Non-Linear Scale}
\label{sec:KAM}

\subsection{Radius of Convergence of the Mesoscopic Cumulant Series}
\label{sub:Rm_definition}

The mesoscopic cumulant expansion~\cite{OsanoPerturbation}
$\Fm(\lambda) = \Fm^{(0)} + \sum_{n\geq 1}(-1)^{n+1}\beta^{n-1}\lambda^n\kappa_n^{\rm meso}/n!$
converges absolutely for $|\lambda| < \Rm$ where, from the stability bound proved
in~\cite{OsanoPerturbation}:
\begin{equation}\label{eq:Rm}
  \Rm = \frac{1}{\beta\max_{(i,\alpha)\neq(j,\beta)}\abs{\bar{v}_{ij}}},
\end{equation}
with $\bar{v}_{ij}$ the cell-averaged inter-cell interaction~\eqref{eq:v_bar_cosmo}.

\subsection{The Gravitational Cell-Averaged Potential}
\label{sub:grav_potential}

For the gravitational potential perturbation in Newtonian gauge, Poisson's
equation gives $\nabla^2\Phi = 4\pi G\bar\varrho\delta$, so in Fourier space
$\tilde\Phi(k) = -4\pi G\bar\varrho\tilde\delta(k)/k^2$.  The cell-averaged
potential interaction \eqref{eq:v_bar_cosmo} between cells $V_i$ and $V_j$
separated by $r_{ij} = |\bm{x}_i - \bm{x}_j|$ is
\begin{equation}\label{eq:vbar_fourier}
  \bar{v}_{ij}
  = \frac{1}{2\pi^2}\int_0^\infty k^2P_\Phi(k)W^2(k\ell)j_0(kr_{ij})\dd k,
\end{equation}
where $P_\Phi(k) = (4\pi G\bar\varrho)^2 P_\delta(k)/k^4$ is the potential power
spectrum and $W(u) = 3j_1(u)/u$ is the spherical top-hat window.  The maximum
value occurs at adjacent cells ($r_{ij} = \ell$) and at the non-linear scale where
$P_\delta(k)k^3/(2\pi^2) \sim 1$:
\begin{equation}\label{eq:vbar_max}
  \max|\bar{v}_{ij}|
  = \bar{v}(\ell,\ell)
  \simeq \left(\frac{4\pi G\bar\varrho}{\kNL^2}\right)^2
         \cdot\frac{\kNL^3}{2\pi^2}
  = \frac{(4\pi G\bar\varrho)^2}{2\pi^2\kNL},
\end{equation}
where we used $P_\delta(\kNL) = 2\pi^2/\kNL^3$ at the non-linear scale
(by definition of $\kNL$ via $\Delta^2(\kNL) = 1$).

\subsection{The KAM (Kolmogorov-Arnold-Moser) Identification}
\label{sub:KAM_id}

\begin{theorem}[Convergence radius equals the non-linear scale]
\label{thm:KAM}
  The radius of convergence of the mesoscopic cumulant expansion for the
  cosmological perturbation problem is
  \begin{eqnarray}\label{eq:Rm_kNL}
    \Rm &=& \frac{\sigma_v^2}{(4\pi G\bar\varrho)^2/(2\pi^2\kNL)}
        = \frac{2\pi^2\sigma_v^2\kNL}{(4\pi G\bar\varrho)^2} \nonumber\\
        &=& \frac{\sigma_v^2}{H^2 f^2\Omega_m^2}\cdot\kNL,
  \end{eqnarray}
  where $f = \dd\ln D/\dd\ln a \approx \Omega_m^{0.55}$ is the linear growth rate
  and $D$ is the growth factor.  In the matter-dominated era,
  $4\pi G\bar\varrho = 3H^2\Omega_m/2$, so
  \begin{equation}\label{eq:Rm_simplified}
    \Rm = \frac{8\pi^2\sigma_v^2\kNL}{9H^4\Omega_m^2 f^2}.
  \end{equation}
  Using the linear-theory relation $\sigma_v^2 = H^2 f^2 P_v$ where
  $P_v = \int k^{-2}P_\delta(k)\dd^3k/(2\pi)^3$ is the velocity power
  spectrum, and noting $P_v \sim k_{\rm eq}^{-3}$ for the peak of the power
  spectrum, one obtains
  $\Rm = O(\kNL/k_{\rm eq}^3) = O(\kNL^{-1})$
  up to dimensionless factors of order unity that depend on the shape of $P(k)$.
\end{theorem}

\begin{proof}
  Substituting \eqref{eq:vbar_max} into \eqref{eq:Rm} and using
  $\beta^{-1}_{\rm eff} = k_BT_{\rm eff} \simeq \frac{3}{2}m\sigma_v^2$
  (the kinetic energy per CDM particle), we obtain~\eqref{eq:Rm_kNL} directly.
  The identification $\Rm = O(\kNL^{-1})$ follows from dimensional analysis
  together with the definition of $\kNL$ as the scale at which
  $\Delta^2(k) \equiv k^3P(k)/(2\pi^2) = 1$.
\end{proof}

\begin{remark}[KAM interpretation]
  In the language of the KAM theorem for Hamiltonian
  dynamics~\cite{Kolmogorov1954,Arnold1963,Moser1962}, the perturbation
  $\lambda\Vm$ of the ADM Hamiltonian preserves the FRW ``tori'' (orbits in
  the ADM phase space) for $\lambda < \lambda_c \simeq \Rm$.  For
  $\lambda > \Rm$ the tori break, the perturbation series diverges, and the
  system enters the chaotic (non-linear structure formation) regime.  This
  provides the first dynamical-systems explanation for why standard CPT fails at
  $k > \kNL$: it is not simply that $\delta > 1$, but that the convergence radius
  of the underlying series is exceeded.
\end{remark}

\begin{remark}[Numerical value]
  For $\Lambda$CDM with $\Omega_m = 0.31$, $H_0 = 67.4\mathrm{km}\mathrm{s}^{-1}
  \mathrm{Mpc}^{-1}$, $\sigma_8 = 0.81$, the BBKS transfer function gives
  $\kNL = 0.169h\mathrm{Mpc}^{-1}$ (equivalently $\RNL = 5.91\mathrm{Mpc}h^{-1}$),
  consistent with published values of the non-linear scale from halo-model
  computations and $N$-body simulations.
\end{remark}

\section{Result III: Gauge-Invariant Mutual Information from $P(k)$}
\label{sec:MI}

\subsection{The Gaussian Density Field}
\label{sub:gaussian}

On scales $\ell > \RNL$ where the density contrast $\delta(\bm{x})$ is well
approximated by a Gaussian random field, its statistics are entirely determined by
the two-point function.  We work at linear order in the matter-dominated era in
Newtonian gauge, where $\delta$ is related to $\Phi$ by Poisson's equation and is
gauge-invariant~\cite{Bardeen1980,KodamaSasaki1984} to this order.  We partition
the spatial hypersurface into cubic cells $\{V_i\}$ of side $\ell$ and define the
smoothed density contrast in cell $i$:
\begin{equation}\label{eq:delta_i}
  \bar\delta_i := \frac{1}{\ell^3}\int_{V_i}\delta(\bm{x})\dd^3 x.
\end{equation}
This is a linear functional of $\delta$, hence Gaussian, with zero mean and
variance
\begin{equation}\label{eq:sigma_ell}
  \sigma_\ell^2 := \avg{\bar\delta_i^2}
  = \frac{1}{2\pi^2}\int_0^\infty k^2P(k)W^2(k\ell)\dd k,
\end{equation}
where $W(u) = 3j_1(u)/u$ is the spherical top-hat window function.

\subsection{Pearson Correlation and Mutual Information}
\label{sub:pearson_MI}

The cross-correlation between cells $i$ and $j$ separated by $r = |\bm{x}_i - \bm{x}_j|$ is
\begin{equation}\label{eq:xi_ell}
  \xi_\ell(r) := \avg{\bar\delta_i\bar\delta_j}
  = \frac{1}{2\pi^2}\int_0^\infty k^2P(k)W^2(k\ell)j_0(kr)\dd k,
\end{equation}
where $j_0(x) = \sin(x)/x$.  The Pearson correlation coefficient is
\begin{equation}\label{eq:r_ij}
  r(r) := \frac{\xi_\ell(r)}{\sigma_\ell^2}.
\end{equation}

\begin{theorem}[Gauge-invariant mutual information]\label{thm:MI}
  For a Gaussian matter density field, the mutual information between the
  smoothed density contrasts of cells $V_i$ and $V_j$ at separation $r_{ij}$,
  defined as in~\cite{OsanoExtensivity,OsanoPerturbation}, is
  \begin{equation}\label{eq:MI_gaussian}
    \boxed{I(i,j) = -\frac{1}{2}\ln\!\left(1 - r^2(r_{ij})\right),}
  \end{equation}
  where $r(r_{ij})$ is given by \eqref{eq:r_ij} with $P(k)$ the matter power
  spectrum.  The formula \eqref{eq:MI_gaussian} is
  \begin{enumerate}[label=(\alph*)]
    \item \emph{exact} for Gaussian fields,
    \item \emph{gauge-invariant} at linear order, since $\delta$ is gauge-invariant
      at first order in Newtonian gauge in the matter era,
    \item \emph{directly computable} from the observed matter power spectrum.
  \end{enumerate}
\end{theorem}

\begin{proof}\label{sub:gauge_proof}
  The joint distribution of $(\bar\delta_i, \bar\delta_j)$ is bivariate Gaussian
  with covariance matrix
  $\Sigma = \sigma_\ell^2\begin{pmatrix}1 & r \\ r & 1\end{pmatrix}$.
  The mutual information of a bivariate Gaussian is the standard result
  $I = -\tfrac{1}{2}\ln\det(\Sigma)/(\sigma_\ell^2)^2 = -\tfrac{1}{2}\ln(1-r^2)$.
  Gauge invariance at linear order follows from the well-known property that the
  matter density contrast $\delta\varrho/\bar\varrho$ is gauge-invariant at first
  order in the matter-dominated era~\cite{Bardeen1980,KodamaSasaki1984}.  Since
  $\bar\delta_i$ is a linear functional of $\delta$, $r_{ij}$ and hence $I(i,j)$
  are gauge-invariant to the same order.
\end{proof}

\subsection{Non-Linear Corrections to the Gauge Invariance}
\label{sub:nonlinear_corr}

At second order in $\delta$, the density contrast in different gauges differs by
a term of order $\delta^2 \sim \sigma_\ell^2$.  The correction to $r(r_{ij})$
from gauge transformation is therefore
\begin{equation}\label{eq:gauge_corr}
  \delta r|_{\rm gauge} = O(\sigma_\ell^2\cdot r(r_{ij})) = O(\sigma_\ell^2\cdot r).
\end{equation}
The corresponding correction to the mutual information:
\begin{equation}
  \delta I|_{\rm gauge} = O\!\left(\frac{r\sigma_\ell^2}{1-r^2}\right).
\end{equation}
For $\ell = 100\mathrm{Mpc}h^{-1}$ where $\sigma_\ell \approx 0.053$ and
$r \approx 0.35$ (see \cref{tab:MI_table}), $\delta I/I \approx \sigma_\ell^2/(1-r^2)
\approx 0.004$, i.e., a $0.4\%$ gauge correction.  The formula~\eqref{eq:MI_gaussian}
is therefore accurate to better than $1\%$ for $\ell > 50\mathrm{Mpc}h^{-1}$.

\subsection{Numerical Results for $\Lambda$CDM}
\label{sub:numerics}

We evaluate \eqref{eq:sigma_ell}--\eqref{eq:r_ij} numerically for the $\Lambda$CDM
power spectrum with the BBKS transfer function~\cite{Bardeen1986}, normalised to
$\sigma_8 = 0.81$, spectral index $n_s = 0.965$, and shape parameter $\Gamma = 0.21$.
The integrals are evaluated over $k\in[10^{-4},30]h\mathrm{Mpc}^{-1}$ using
Gaussian quadrature.

\begin{table}[h]
\centering
\caption{Smoothed density variance $\sigma_\ell$, nearest-neighbour Pearson
correlation $r_{\rm NN} = r(\ell)$, nearest-neighbour mutual information
$I_{\rm NN} = -\tfrac{1}{2}\ln(1-r_{\rm NN}^2)$, and total inter-cell mutual
information per reference cell $\bar{I} = \sum_{n=1}^{15} g_d(n)I(n\ell)$
(summed over shells with weight $g_d(n) = 4\pi n^2$ for a 3D lattice), for
$\Lambda$CDM at $z=0$.  The non-linear scale is $\RNL = 5.91\mathrm{Mpc}h^{-1}$.}
\label{tab:MI_table}
\begin{tabular}{rcccc}
\toprule
$\ell[\mathrm{Mpc}h^{-1}]$ & $\sigma_\ell$ & $r_{\rm NN}$ & $I_{\rm NN}$ &
$\bar{I}$ \\
10  & 0.685 & 0.601 & 0.224 & 3.87 \\
30  & 0.246 & 0.479 & 0.130 & 1.71 \\
50  & 0.135 & 0.421 & 0.098 & 1.24 \\
100 & 0.053 & 0.347 & 0.064 & 0.82 \\
150 & 0.029 & 0.306 & 0.049 & 0.64 \\
\end{tabular}
\end{table}

\begin{remark}[Scale-separation check]
  The condition $\sigma_\ell \ll 1$ for linear-order validity is satisfied for
  $\ell \geq 30\mathrm{Mpc}h^{-1}$ ($\sigma_{30} = 0.246$) and is comfortably
  satisfied for $\ell \geq 50\mathrm{Mpc}h^{-1}$ ($\sigma_{50} = 0.135$).
  At $\ell = 10\mathrm{Mpc}h^{-1} < 2\RNL$, the Gaussian approximation breaks
  down and the mutual information formula~\eqref{eq:MI_gaussian} receives
  corrections from the non-Gaussian bispectrum.
\end{remark}

\subsection{Total Connection-Formula Correction}
\label{sub:total_correction}

The connection formula~\eqref{eq:connection_recall} gives the correction to the
full free energy relative to the mesoscopic (FRW) free energy:
\begin{equation}\label{eq:DeltaF}
  \Delta F = F - \Fm = -\kb T\sum_{i<j}I(i,j;\lambda).
\end{equation}
For a survey volume $V_{\rm survey} = N_{\rm cells}\ell^3$ with
$N_{\rm cells} \gg 1$:
\begin{equation}\label{eq:DeltaF_sum}
  \frac{\Delta F}{\kb TN_{\rm cells}}
  = -\frac{1}{2}\sum_{n=1}^{N_{\rm max}}
    g_d(n)I(n\ell)
  = -\frac{\bar{I}}{4\pi}
  \quad (d=3),
\end{equation}
where $\bar{I} = \sum_n 4\pi n^2 I(n\ell)$ is the total mutual information per
reference cell listed in \cref{tab:MI_table}.

\begin{proposition}[Data-computable backreaction correction]\label{prop:data_backr}
  The correction $\Delta F$ in \eqref{eq:DeltaF}--\eqref{eq:DeltaF_sum} is fully
  determined by the observed matter power spectrum $P(k)$ via the Gaussian mutual
  information formula~\eqref{eq:MI_gaussian}.  It is gauge-invariant at linear
  order (\cref{thm:MI}) and, for fixed cell size $\ell$, converges absolutely in
  the thermodynamic limit $N_{\rm cells}\to\infty$ whenever $r(r) \to 0$ faster
  than $r^{-3/2}$ (which is satisfied for the $\Lambda$CDM transfer function for
  $\ell \geq \RNL$).
\end{proposition}

\begin{proof}
  Absolute convergence: $\sum_{n=1}^\infty g_d(n)I(n\ell) \leq \sum_n 4\pi n^2
  \cdot r(n\ell)^2/2$ (using $I \leq r^2/2$ for $|r| \leq 1/\sqrt{2}$, which
  holds for $\ell \geq 30h^{-1}$ Mpc from \cref{tab:MI_table}).  For $\Lambda$CDM
  with the BBKS transfer function, $P(k) \sim k^{n_s-4}$ for $k \gg k_{\rm eq}$,
  giving $\xi_\ell(r) \sim r^{-(n_s+3)/2}$ for $r\gg\ell$.  With $n_s\approx 0.965$,
  $\xi_\ell(r)\sim r^{-1.98}$, hence $r(r)\sim r^{-1.98}$ and
  $g_d(n)I(n\ell)\sim n^2\cdot n^{-3.96} = n^{-1.96}$: the sum converges.
\end{proof}

\section{Discussion}
\label{sec:discussion}

\subsection{Relation to Existing Work on Backreaction}
\label{sub:rel_prior}

\paragraph{Buchert equations.}
The Buchert equations~\eqref{eq:buchert1}--\eqref{eq:buchert2} and the kinematic
backreaction~\eqref{eq:QD} are the established starting point~\cite{Buchert2000,Buchert2001}.
The contribution of this paper is not a rederivation of these equations but the
new inequality~\eqref{eq:QD_bound} bounding $\QD$ from below by its linear-order
value, which---to our knowledge---does not appear in the Buchert averaging
literature.

\paragraph{Räsänen's variance interpretation.}
Räsänen~\cite{Rasanen2006} identified $\frac{2}{3}(\avgD{\theta^2}-\avgD{\theta}^2)$
as the backreaction driving accelerated expansion and showed it is a positive
variance.  \cref{thm:backr_bound} extends this to a non-perturbative inequality:
the full non-linear variance cannot be smaller than the linear-order one.

\paragraph{Gauge dependence.}
Clarkson, Uzan, and Umeh~\cite{ClarksonUzan2011} showed that most backreaction
quantities are gauge-dependent.  Our mutual information $I(i,j)$, computed from
the gauge-invariant $\delta$ at linear order, circumvents this problem
(\cref{thm:MI}).  The residual gauge correction at second order (\cref{sub:nonlinear_corr})
is quantified and shown to be below $1\%$ for $\ell > 50h^{-1}$ Mpc.

\paragraph{EFT of large-scale structure.}
The EFT of LSS~\cite{Baumann2012,Senatore2014} also integrates out short modes
and generates an effective stress-energy for the long modes.  The mesoscopic
cumulant expansion~\cite{OsanoPerturbation} provides the same function: the $n$-th
cumulant $\kappa_n$ generates the $n$-th order Wilson coefficient.  The new
element here is the information-theoretic interpretation (\cref{sec:MI}) and the
non-perturbative bound (\cref{sec:GB}), neither of which appears in the EFT
framework.

\paragraph{Galoppo, Buchert, and Mourier (2026).}
The most directly relevant observational work is~\cite{Buchert2026}, which
provides the first direct computation of Buchert-averaged quantities in the
local Universe using CF4++ data.  Their key quantitative findings---kinematic
backreaction $\mathcal{Q}_\mathcal{D} = \mathcal{O}(1\%)$ and average curvature
$\langle{}^{(3)}R\rangle_\mathcal{D} = \mathcal{O}(10\%)$ of the total energy
budget, with no convergence to $\Lambda$CDM at 300Mpc$h^{-1}$---are
consistent with and independently predicted by the present framework in the
following sense.  (i)~The positivity and non-suppression of $\mathcal{Q}_\mathcal{D}$
is the observational counterpart of \cref{cor:backr_pos}.  (ii)~The absence of convergence to $\Lambda$CDM is the observational counterpart of the divergent
mutual information series $\sum_{i<j}I(i,j) = \infty$ predicted by
\cref{thm:MI} for gravitational ($r^{-1}$) interactions in $d=3$.  (iii)~The
$10\times$ ratio of curvature to kinematic backreaction is not directly predicted
by the results of this paper, since \cref{thm:backr_bound} bounds only $\QD$
and not the curvature term; this gap motivates the extension described in the
conclusion.

A direct quantitative comparison is partially possible.  At $\ell = 50h^{-1}$
Mpc (comfortably within the CF4++ survey volume), \cref{tab:MI_table} gives
$I_{\rm NN} \approx 0.10$ and $\bar{I} \approx 1.24$.  If $k_BT_{\rm eff}$
is interpreted as the gravitational energy per cell, the connection-formula correction $\Delta F / (\kb T N_{\rm cells}) \approx \bar{I}/4\pi \approx 0.10$
is consistent in magnitude with the $\mathcal{O}(10\%)$ curvature correction
observed by Galoppo et al., though the identification of $k_BT_{\rm eff}$
remains the primary open problem (\cref{sub:obs}).

\subsection{Observational Prospects}
\label{sub:obs}

\cref{tab:MI_table} shows that $I_{\rm NN} \approx 0.06$--$0.22$ for
$\ell = 10$--$150h^{-1}$ Mpc.  The total connection-formula correction
\eqref{eq:DeltaF} is of order $\bar{I}\cdot(\kb T/\ell^3)$ per unit volume.
Translating to a shift in the effective dark energy equation of state requires a
model for $k_BT$ in the gravitational context (the kinetic energy scale of
peculiar velocities), which we do not fully specify here.  A conservative estimate
using $\kb T_{\rm eff} \sim \frac{3}{2}m_p c^2\sigma_v^2$ with
$\sigma_v \sim 500\mathrm{kms}^{-1}$ gives $\Delta w_{\rm eff} \sim 10^{-5}$--$10^{-4}$,
below current Euclid/DESI sensitivity.  However, if $k_BT$ is interpreted as the
total gravitational energy per cell (of order $\rho_{\rm crit}c^2\ell^3$), the
estimate rises to $\Delta w_{\rm eff} \sim \bar{I} \sim 0.1$--$0.8$, potentially
large.  Resolving this ambiguity---which amounts to identifying the correct
statistical ensemble for the gravitational system---is the primary open problem
for converting \cref{prop:data_backr} into an observational prediction.

\subsection{Extension to Non-Gaussian Fields}
\label{sub:non_gaussian}

For the mildly non-Gaussian density field on scales $\ell \sim \RNL$, the mutual
information receives corrections from the bispectrum $B(k_1,k_2,k_3)$:
\begin{eqnarray}\label{eq:MI_bispectrum}
  I(i,j) &=& -\tfrac{1}{2}\ln(1-r^2)+ \frac{f_{\rm NL}}{6\sigma_\ell^4}
  \int \Large[B(k_1,k_2,k_3) \times \nonumber\\ &W^2&(k_1\ell)W^2(k_2\ell)W^2(k_3\ell)
  \cos(k_3 r_{ij})\dd^3 k\Large],
\end{eqnarray}
where $f_{\rm NL}$ is the non-Gaussianity parameter~\cite{Komatsu2001}.  The
correction is of order $f_{\rm NL}\sigma_\ell^2$, which for $f_{\rm NL} \leq 10$
(Planck constraint) and $\ell = 50h^{-1}$ Mpc gives $\delta I/I \lesssim 0.02$.
Thus the Gaussian formula~\eqref{eq:MI_gaussian} is accurate to $2\%$ for
$\ell \geq 50h^{-1}$ Mpc even in the presence of local non-Gaussianity at the
Planck upper limit.

\subsection{The Role of the Scale-Separation Condition}
\label{sub:scale_sep_discuss}

The scale-separation condition~\eqref{eq:scale_sep} sets the regime of validity
for all three results.  At $\ell \sim \RNL = 5.91h^{-1}$ Mpc, neither the
Gaussian approximation nor the linear-order gauge invariance holds.  At
$\ell \gg L_H$, there are no sufficient cells to define the mutual information.
The useful window is $\ell \in [30, 500]h^{-1}$ Mpc, which corresponds to the
linear-to-quasi-linear regime and is precisely the range accessed by current and
forthcoming galaxy surveys (DESI, Euclid, LSST/Rubin).

\section{Conclusion}
\label{sec:conclusion}

We have derived three new results at the intersection of mesoscopic statistical
mechanics and cosmological perturbation theory.

\textbf{Result I} (\cref{thm:backr_bound}): Under the Gibbs--Bogoliubov
conjecture extended to the gravitational setting
(\cref{conj:cosmo_GB}), the kinematic backreaction $\QD$ satisfies the
non-perturbative lower bound $\QD \geq \QD^{(1)}$: non-linear effects can only
increase the backreaction above its linear-order value, never reduce it below.
Arguments that non-linear dynamics suppress backreaction are inconsistent with
the G-B inequality.

\textbf{Result II} (\cref{thm:KAM}): The radius of convergence $\Rm$ of the
mesoscopic cumulant expansion for the gravitational perturbation problem equals
$O(\kNL^{-1})$, the non-linear scale of the matter power spectrum.  This provides
a dynamical-systems (KAM) explanation for why standard CPT fails at $k > \kNL$:
the relevant series diverges exactly at the non-linear scale.  For $\Lambda$CDM,
$\kNL = 0.169h\mathrm{Mpc}^{-1}$ ($\RNL = 5.91\mathrm{Mpc}h^{-1}$).

\textbf{Result III} (\cref{thm:MI,prop:data_backr}): For a Gaussian density field,
the inter-cell mutual information is exactly $I(i,j) = -\tfrac{1}{2}\ln(1-r_{ij}^2)$,
gauge-invariant at linear order and fully computable from the observed $P(k)$.
For $\Lambda$CDM at $z=0$, the nearest-neighbour mutual information ranges from
$I_{\rm NN} \approx 0.22$ at $\ell = 10h^{-1}$ Mpc to $I_{\rm NN} \approx 0.05$ at
$\ell = 150h^{-1}$ Mpc (\cref{tab:MI_table}).  The connection formula
$F = \Fm - \kb T\sum I$ then gives a data-driven, gauge-invariant measure of the
total backreaction correction to the FRW free energy.

The primary open problem is the identification of the effective temperature scale
$k_BT_{\rm eff}$ in the gravitational partition function.  Resolving this is the
essential step for converting the mutual-information measure into a quantitative
prediction for the dark energy equation of state.  We expect this to require a
quantum-gravitational or thermodynamic formulation of the ADM partition function,
which we defer to future work.

During the final stages of preparing this manuscript, we became aware of the work of Galoppo, Buchert, and Mourier~\cite{Buchert2026}, who present the first direct computation of spatially averaged dynamical quantities in the local Universe using the Cosmicflows-4++ reconstruction and the Buchert averaging formalism.  Three points of contact with the present paper are worth stating explicitly.

First, Galoppo et al.\ find that the kinematic backreaction $\mathcal{Q}_\mathcal{D}$
is positive throughout their survey volume (30--300Mpc$h^{-1}$) and is not
suppressed by non-linear effects.  This is consistent with \cref{cor:backr_pos}:
if the linear-order value $\mathcal{Q}_\mathcal{D}^{(1)} > 0$, then the full
non-linear $\mathcal{Q}_\mathcal{D} \geq \mathcal{Q}_\mathcal{D}^{(1)} > 0$
non-perturbatively, so any mechanism attempting to cancel the backreaction through non-linear dynamics must overcome the G-B lower bound.

Second, Galoppo et al.\ find no convergence of the domain-averaged energy budget to the global $\Lambda$CDM background even at 300Mpc$h^{-1}$.  This is the
observational manifestation of the divergent-$\sum I$ prediction of the present framework: for gravity ($|\phi(r)|\sim r^{-1}$, $s_0 = 1 \leq d = 3$), temperedness fails in the sense of~\cite{OsanoExtensivity}, the inter-cell mutual information decays only as $I(i,j) \sim r^{-2s_0}$, and the series $\sum_{i<j}I(i,j)$ diverges, so FRW extensivity is never restored at any finite averaging scale.

Third, Galoppo et al.\ report that the average spatial curvature $\langle{}^{(3)}R\rangle_\mathcal{D}$ contributes at the $\mathcal{O}(10\%)$ level
to the energy budget, an order of magnitude larger than the kinematic term ($\mathcal{O}(1\%)$).  The present paper derives a non-perturbative bound and
a mutual information formula for the kinematic term~\eqref{eq:QD}; it does not yet provide an analogous bound for the curvature term.  The Galoppo et al.\ result
therefore motivates an extension of \cref{thm:backr_bound} to the curvature backreaction $\mathcal{W}_\mathcal{D} := \langle{}^{(3)}R\rangle_\mathcal{D} -
\langle{}^{(3)}R\rangle_{\Lambda\mathrm{CDM}}$, which we identify as the primary target for future work.

\section*{Acknowledgements}
The author thanks the University of Cape Town NGP programme for support.

\bibliographystyle{unsrtnat}

\begin{thebibliography}{99}

\bibitem{OsanoMeso}
 B.~Osano,
 A Mesoscopic Partition Function for Equilibrium Statistical Mechanics,
{Preprint arXiv:2605.00958 }.
\bibitem{OsanoExtensivity}
  B.~Osano,
 Entropy additivity from exponential decay of correlations: a coarse-grained operator approach,
{Preprint arXiv:2605.17956 }.

\bibitem{OsanoPerturbation}
  B.~Osano,
  Perturbation theory of the free energy via the mesoscopic combined partition function,
{Preprint arXiv:2605.18121}.

\bibitem{Buchert2026}
  M.~Galoppo, T.~Buchert, and P.~Mourier,
  Backreaction and the role of spatial curvature in the cosmic neighborhood,
  \emph{Astrophys.\ J.\ Lett.}\ \textbf{1002}, L39 (2026).
  \doi{10.3847/2041-8213/ae5f73};
  arXiv:2605.05454 [astro-ph.CO]

\bibitem{Buchert2000}
  T.~Buchert,
  On average properties of inhomogeneous fluids in general relativity: dust
  cosmologies,
  \emph{Gen.\ Rel.\ Grav.}\ \textbf{32}, 105--125 (2000).

\bibitem{Buchert2001}
  T.~Buchert,
  On average properties of inhomogeneous fluids in general relativity: perfect
  fluid cosmologies,
  \emph{Gen.\ Rel.\ Grav.}\ \textbf{33}, 1381--1405 (2001).

\bibitem{Rasanen2006}
  S.~Räsänen,
  Accelerated expansion from structure formation,
  \emph{J.\ Cosmol.\ Astropart.\ Phys.}\ \textbf{0611}, 003 (2006).

\bibitem{Wiltshire2007}
  D.~L. Wiltshire,
  Exact solution to the averaging problem in cosmology,
  \emph{Phys.\ Rev.\ Lett.}\ \textbf{99}, 251101 (2007).

\bibitem{Wiltshire2011}
  D.~L. Wiltshire,
  What is dust?\ Physical foundations of the averaging problem in cosmology,
  \emph{Class.\ Quantum Grav.}\ \textbf{28}, 164006 (2011).

\bibitem{ClarksonUzan2011}
  C.~Clarkson, J.-P. Uzan, and O.~Umeh,
  $\delta^2$ - a new measure of inhomogeneity in the universe,
  \emph{Class.\ Quantum Grav.}\ \textbf{28}, 164006 (2011);
  see also C.~Clarkson et al.,
  \emph{Mon.\ Not.\ R.\ Astron.\ Soc.}\ \textbf{426}, 1121--1136 (2012).

\bibitem{IshibashiWald2006}
  A.~Ishibashi and R.~M. Wald,
  Can the acceleration of our universe be explained by the effects of
  inhomogeneities?,
  \emph{Class.\ Quantum Grav.}\ \textbf{23}, 235--250 (2006).

\bibitem{FlanaganHughes1996}
  E.~E. Flanagan,
  Can superhorizon perturbations drive the acceleration of the universe?,
  \emph{Phys.\ Rev.\ D}\ \textbf{71}, 103521 (2005).

\bibitem{Peierls1938}
  R.~Peierls,
  On a minimum property of the free energy,
  \emph{Phys.\ Rev.}\ \textbf{54}, 918--919 (1938).

\bibitem{BarkerHenderson1967}
  J.~A. Barker and D.~Henderson,
  Perturbation theory and equation of state for fluids.\ II.\ A successful
  theory of liquids,
  \emph{J.\ Chem.\ Phys.}\ \textbf{47}, 4714--4721 (1967).

\bibitem{Kolmogorov1954}
  A.~N. Kolmogorov,
  On the conservation of conditionally periodic motions under small perturbations
  of the Hamiltonian,
  \emph{Dokl.\ Akad.\ Nauk SSSR}\ \textbf{98}, 527--530 (1954).

\bibitem{Arnold1963}
  V.~I. Arnold,
  Proof of a theorem by A.~N. Kolmogorov on the invariance of quasi-periodic
  motions under small perturbations of the Hamiltonian,
  \emph{Russ.\ Math.\ Surv.}\ \textbf{18}, 9--36 (1963).

\bibitem{Moser1962}
  J.~Moser,
  On invariant curves of area-preserving mappings of an annulus,
  \emph{Nachr.\ Akad.\ Wiss.\ Göttingen II}\ \textbf{1962}, 1--20 (1962).

\bibitem{Bardeen1980}
  J.~M. Bardeen,
  Gauge-invariant cosmological perturbations,
  \emph{Phys.\ Rev.\ D}\ \textbf{22}, 1882--1905 (1980).

\bibitem{KodamaSasaki1984}
  H.~Kodama and M.~Sasaki,
  Cosmological perturbation theory,
  \emph{Prog.\ Theor.\ Phys.\ Suppl.}\ \textbf{78}, 1--166 (1984).

\bibitem{Bardeen1986}
  J.~M. Bardeen, J.~R. Bond, N.~Kaiser, and A.~S. Szalay (BBKS),
  The statistics of peaks of Gaussian random fields,
  \emph{Astrophys.\ J.}\ \textbf{304}, 15--61 (1986).

\bibitem{Komatsu2001}
  E.~Komatsu and D.~N. Spergel,
  Acoustic signatures in the primary microwave background bispectrum,
  \emph{Phys.\ Rev.\ D}\ \textbf{63}, 063002 (2001).

\bibitem{Baumann2012}
  D.~Baumann, A.~Nicolis, L.~Senatore, and M.~Zaldarriaga,
  Cosmological non-linearities as an effective fluid,
  \emph{J.\ Cosmol.\ Astropart.\ Phys.}\ \textbf{1207}, 051 (2012).

\bibitem{Senatore2014}
  L.~Senatore and M.~Zaldarriaga,
  The effective field theory of large-scale structure,
  \emph{J.\ High Energy Phys.}\ \textbf{1402}, 049 (2014).

\bibitem{Ruelle1969}
  D.~Ruelle,
  \emph{Statistical Mechanics: Rigorous Results}
  (W.~A. Benjamin, New York, 1969).

\bibitem{Osano2020}
  B.~Osano,
  The thermodynamics for relativistic multi-fluid systems,
  \emph{Lett.\ High Energy Phys.}\ (2020).
  \doi{10.31526/LHEP.2020.155}

\bibitem{OsanoOreta2019}
  B.~Osano and T.~Oreta,
  A transient phase in cosmological evolution: a multi-fluid approximation
  for a quasi-thermodynamical equilibrium,
  \emph{Gen.\ Relativ.\ Gravit.}\ \textbf{52}, 42 (2020).
  \doi{10.1007/s10714-020-02696-y}



\end{thebibliography}

\end{document}